\documentclass[conference]{IEEEtran}
\IEEEoverridecommandlockouts

\usepackage{cite}
\usepackage{amsmath,amssymb,amsfonts,amsthm}
\usepackage{algorithmic}
\usepackage{graphicx}
\usepackage{textcomp}
\usepackage{xcolor}
\usepackage[normalem]{ulem}
\newtheorem{theorem}{Theorem}
\usepackage{subfigure}
\def\BibTeX{{\rm B\kern-.05em{\sc i\kern-.025em b}\kern-.08em
    T\kern-.1667em\lower.7ex\hbox{E}\kern-.125emX}}

\usepackage[top=0.8in, bottom=1.05in, left=0.682in, right=0.682in]{geometry}
    
\begin{document}

\title{A Remark on Downlink Massive Random Access}

\author{
\IEEEauthorblockN{Yuchen Liao and Wenyi Zhang}
\IEEEauthorblockA{\textit{Department of Electronic Engineering and Information Science} \\
\textit{University of Science and Technology of China}\\
\textit{Hefei, China}\\
Email: wenyizha@ustc.edu.cn}
\thanks{The work was supported in part by the National Key Research and Development Program of China (2025YFF0514401) and the National Natural Science Foundation of China (62231022).}
}

\maketitle

\begin{abstract}
In downlink massive random access (DMRA), a base station transmits messages to a typically small subset of active users, selected randomly from a massive number of total users. Explicitly encoding the identities of active users would incur a significant overhead scaling logarithmically with the number of total users. Recently, via a random coding argument, Song, Attiah and Yu have shown that the overhead can be reduced to within some upper bound irrespective of the number of total users. In this remark, recognizing that the code design for DMRA is an instance of covering arrays in combinatorics, we show that there exists deterministic construction of variable-length codes that incur an overhead no greater than $1 + \log_2 e$ bits.
\end{abstract}

\section{Introduction}\label{sec:intro}

Suppose that a transmitter (a.k.a. base station) wishes to transmit messages to $k$ active receivers (a.k.a. users), selected randomly from $n$ total users, where typically $k \ll n$. Each active user only knows that it is set to receive its own message, but is unaware of which other $k - 1$ active users are. For this downlink massive random access (DMRA) setup, if the base station explicitly encodes the identities of the $k$ active users, an overhead of $k \log_2 n$ bits is incurred.\footnote{This corresponds to a scheme in which a string of length $\log_2 n$ bits is used for representing the identity of each of the $k$ active users. A slightly more efficient scheme is to encode the identities of the $k$ active users altogether, using $\log_2 \binom{n}{k}$ bits, which is also close to $k \log_2 n$ for $k \ll n$.}

Is such an overhead necessary for DMRA? In a recent work \cite{song25tit}, via a random coding argument, Song, Attiah and Yu have shown that the overhead can be dramatically reduced, to within some upper bound irrespective of $n$.

Such a seemingly counterintuitive result can be explained by considering a case where the $k$ active users are drawn uniformly randomly from the $\binom{n}{k}$ possible combinations, and their messages are independent and identically distributed (i.i.d.).\footnote{The results in \cite{song25tit} also cover the more general case where message distributions are exchangeable. In the current paper we focus on the i.i.d. case for simplicity.} For simplicity, suppose that each active user's message is a single bit. Generate a codebook as a sequence of length-$n$ binary vectors, indexed as $1, 2, \ldots$. Given $k$ active users and their associated messages, search sequentially among the vectors in the generated codebook until finding a vector whose elements at positions corresponding to the $k$ active users exactly match these $k$ active users' messages. Finally, use a variable-length uniquely decodable lossless code (e.g., Shannon code) to encode the index of the found vector. For example, with codebook $\{0000, 1111, 0101, 1010, 1100, \ldots\}$, when users $1$ and $3$ are active with messages $1$ and $0$ respectively, one needs to search for a vector of pattern 1X0X, and the first such vector is $1100$ with index $5$.

Under a random coding argument, one considers the ensemble performance averaged over all codebooks. Let the vectors in the random codebook be mutually independent, each consisting of $n$ i.i.d. Bernoulli$(1/2)$ random variables. Therefore, for given $k$ active users and their associated messages, each vector in the random codebook matches them with probability $2^{-k}$, and the first matched index is a geometric random variable with mean $2^k$. Such a geometric random variable has an entropy slightly larger than $k$ bits, thereby corresponding to a variable-length lossless code of rate independent of $n$.

The ensemble analysis in \cite{song25tit} basically describes the averaged performance as we randomly generate codebooks, but in reality we typically need to work with a fixed codebook repeatedly. Therefore, it still remains to answer how to construct a codebook that can approach the ensemble performance.

In this work, we point out that the DMRA model is intimately connected with a classical topic in combinatorics known as covering arrays. Briefly, a covering array, usually denoted by $\mathsf{CA}(M; k, n, q)$, is an $M \times n$ array, whose symbols are over an alphabet of size $q$, so that for every possible size-$k$ subset of $\{1, 2, \ldots, n\}$, every $q^k$ possible combinations of symbols appears in at least one row of the array. Clearly, when each active user's message is over a size-$q$ alphabet, $\mathsf{CA}(M; k, n, q)$ is a natural choice of DMRA codebook.

In Table \ref{tab:concise_codebook_en}, we illustrate an instance of $\mathsf{CA}(5; 2, 4, 2)$ used as a DMRA codebook. 

\begin{table}[htbp]
    \vspace{-0.7em}
    \centering
    \caption{Example of $\mathsf{CA}(5; 2, 4, 2)$ as DMRA codebook}
    \label{tab:concise_codebook_en}
    \begin{tabular}{|c|c|c|}
        \hline
        \textbf{Index} & \textbf{Vector} & \textbf{Covered patterns} \\ 
        \hline
        1 & (0,0,0,0) & 00XX 0X0X 0XX0 X00X X0X0 XX00 \\ 
        \hline
        2 & (1,1,1,0) & 11XX 1X1X 1XX0 X11X X1X0 XX10 \\ 
        \hline
        3 & (1,1,0,1) & 1X0X 1XX1 X10X X1X1 XX01 \\ 
        \hline
        4 & (1,0,1,1) & 10XX X01X X0X1 XX11 \\ 
        \hline
        5 & (0,1,1,1) & 01XX 0X1X 0XX1 \\ 
        \hline
    \end{tabular}
    \vspace{-0.7em}
\end{table}

The numbers of covered patterns over the indices are $6, 6, 5, 4, 3$, respectively, summing up to $\binom{4}{2} 2^2 = 24$ possible source-message patterns. The probabilities of the indices $\{1,2,3,4,5\}$ are thus $\{1/4,1/4,5/24,1/6,1/8\}$, respectively, leading to an entropy of $2.28$ bits, only slightly higher than the net payload of $2$ bits for two active users. The Huffman code for encoding the indices has an expected length of $2.29$ bits. In contrast, explicitly encoding the identities of two active users incurs an overhead of $2 \log_2 4 = 4$ bits, resulting in a total of $2 + 4 = 6$ bits, significantly higher than $2.29$ bits when employing the covering array in Table \ref{tab:concise_codebook_en}.

In this remark, we show that for general DMRA parameters, there exist deterministic covering array constructions such that the resulting variable-length codes incur an overhead no greater than $1 + \log_2 e$ bits, regardless of the number of total users. This result is established by exploiting a property of covering arrays regarding how the number of uncovered patterns evolves as the index increases. Roughly speaking, there exist covering arrays for which the number of uncovered patterns decays at least geometrically fast with index, and greedy algorithms suffice to construct such covering arrays.

The remaining part of this remark is organized as follows. Section \ref{sec:rel.} briefly reviews related works on massive access and covering arrays. Section \ref{sec:pro.} describes the problem formulation of DMRA. Section \ref{sec:per.} establishes the performance of DMRA codebook constructions based on covering arrays. Section \ref{sec:numerical} presents some numerical results for illustration. Finally, Section \ref{sec:conclusion} concludes this remark.

\section{Related Works}\label{sec:rel.}

This section reviews the literature from two key domains underpinning our work: massive machine-type communications (mMTC) in wireless communication, and the combinatorial theory of covering arrays.

\subsection{Massive Access}

Massive access, also termed mMTC, is a pivotal scenario in 5G and beyond wireless communication, supporting a vast network of devices with sporadic traffic, where only a typically small subset is active simultaneously\cite{chen20jsac}. This sporadic nature of traffic presents unique challenges for resource allocation and protocol design.

Substantial research efforts have been devoted to uplink massive access scenarios. Key considerations therein include active user identification and unsourced random access. The former has been treated in many-access channels (MnAC) (see, e.g., \cite{chen17tit, robin21isit, ravi21tit}), while the latter has been proposed in \cite{polyanskiy17isit} and studied subsequently (see, e.g., \cite{fengler19arxiv1, gao22tit}).

In contrast, the downlink massive access scenario — where a base station transmits to a typically small random subset of active users — has received comparatively less attention \cite{song25tit,kang21tit,song22isit,kalor22tc}. Our work adopts the DMRA model proposed in \cite{song25tit}, and studies it through the lens of covering arrays, as to be reviewed in the next subsection. 

\subsection{Covering Arrays}

Covering arrays has a rich history in combinatorics. Their applications have been mainly in software and hardware testing \cite{sloane93jcd,colbourn04lm,torres13iss,lawrence11tejoc}, where they are used to generate test suites that cover all interactions among a fixed number of system factors.

For apparent reasons, a central combinatorial problem regarding covering arrays has been to determine the smallest size of a covering array for a given tuple of parameters. That is, we want to find the minimum $M$, denoted by $M(k, n, q)$, such that $\mathsf{CA}(M; k, n, q)$ exists. For general $(k, n, q)$, the exact value of $M(k, n, q)$ remains an open problem in general, except for the special case of $k = q = 2$ \cite{katona, kleitman_and_spencer}. 
Extensive research has explored upper bounds for $M(k,n,q)$ under various parameter settings \cite{torres12ins, torres18dcc, torres17optl, torres18mpe}. The currently best-known results have been cataloged in repository \cite{colbourn25online}. Furthermore, asymptotic bounds have been explored as $n$ grows for fixed $k$ and $q$, providing insights into the scalability of covering arrays \cite{gargano93gac,bounds_u, sarkar17siam, francetic17jcd}.

A diverse array of algorithms have been developed for constructing covering arrays (see \cite{torres19access} and references therein), which can generally be categorized into three types: combinatorial methods, searching methods, and hybrid methods. Combinatorial methods offer rapid generation and currently represent the most viable approach for large values of $(k, n, q)$, though they are limited to restricted parameter combinations. Searching methods, conversely, can handle arbitrary values of $(k, n, q)$ but face limitations imposed by computational resources. In the next section, our approach to DMRA coding will be based on a searching method, namely row-by-row greedy construction algorithms \cite{cohen97aetg,bryce07density}, which do not necessarily yield the minimum $M(k, n, q)$, but suffice for our purposes.

\section{Problem Formulation}\label{sec:pro.}

This section formalizes the DMRA coding problem, which has already been outlined in Section \ref{sec:intro}.

A \textbf{user-pattern vector} $\uline{S} \in {\{0,1\}}^n$ indicates the activation status of users, in which:
\begin{equation}
    S_i =
    \left\{
    \begin{array}{ll}
         0, \ \ \ \ \ \ &\text{if\ user\ }i\text{\ is\ inactive,}\\
         1, \ \ \ \ \ \ &\text{if\ user\ }i\text{\ is\ active.}
    \end{array}
    \right.
\end{equation}
It follows that $\uline{S}$ satisfies $\sum_{i=1}^{n} S_i = k$, ensuring exactly $k$ active users.

A \textbf{message vector} $\uline{W} \in {\{0,1,\ldots,q-1,\phi \}}^n$ contains the messages intended for the users at the corresponding indices:
\begin{equation}
    W_i =
    \left\{
    \begin{array}{ll}
         \sim\text{uniform\ in\  }\{0,1,\ldots,q-1\}, \ \ \ &\text{if\  }S_i = 1,\\
         \phi, \ \ \ &\text{if\ }S_i = 0.
    \end{array}
    \right.
\end{equation}
where $\phi$ indicates that no message (i.e., void) is intended for an inactive user. The message for each active user is assumed to be uniformly distributed over the message alphabet, and all these messages are mutually independent. This will simplify the analysis in Section \ref{sec:per.}.

We focus on the efficient representation of identities and messages of active users, and hence we frame the DMRA coding problem as a lossless source coding problem. The effects of noisy channels are not considered, as in practice channel coding can be applied additionally to encapsulate the lossless source codewords.

A DMRA code has a codebook $\mathcal{C}$ which is apriori known to the base station and all users, consisting of $M$ possibly variable-length binary codewords:
\begin{equation}
    \mathcal{C}=(c{(1)},c{(2)},\ldots,c{(M)}).
\end{equation}

\textbf{Encoding:} 
An encoder is a mapping $g$ that maps every $(\uline{s}, \uline{w})$-pair to an index $m$ among $\{1, \ldots, M\}$:
\begin{equation}
    m = g(\uline{s}, \uline{w}),
\end{equation}
and $c(m)$ is used for representing $(\uline{s}, \uline{w})$, which will be termed a pattern subsequently. It is crucial to note that this mapping need not be injective; in other words, several different patterns may be mapped to the same index $m$.

\textbf{Decoding:} Each active user $i$ decodes $c(m)$ to retrieve its intended message $w_i$; inactive users do not engage in decoding, of course. The decoder mapping for user $i$ can thus be written as:
\begin{equation}
    \hat{w}_i = f(c(m), s_i).
\end{equation}
When $s_i = 0$, we immediately have $\hat{w}_i = \phi$. The key is to ensure that when $s_i = 1$, $\hat{w}_i = w_i$. 

We denote the length of $c(m)$ by $\ell(m)$.

The random pattern $(\uline{S}, \uline{W})$ induces a probability distribution on the encoder output, $g(\uline{S}, \uline{W})$, over the alphabet $\{1, \ldots, M\}$. The coding efficiency can therefore be measured by the expected codeword length:
\begin{equation}
    \bar{\ell} = \mathbb{E}[\ell(g(\uline{S}, \uline{W}))].
\end{equation}

\section{Main Result}\label{sec:per.}

This section develops a general upper bound for the optimal $\bar{\ell}$, as shown in the following theorem.

\begin{theorem}\label{thm:main}
    For any $(n, k, q)$, $k \leq n$, $q \geq 2$, there exist deterministic DMRA codes such that the expected codeword length satisfies:
    \begin{equation}\label{var.bound}
        \bar{\ell} < k \log q + 1 + \log e.
    \end{equation}
\end{theorem}

\begin{proof}

    As hinted by the example in the introduction, we will adopt covering arrays to construct the DMRA codebook. In a nutshell, we aim at constructing some $\mathsf{CA}(M; k, n, q)$. During encoding, given a $(\uline{s}, \uline{w})$-pair, i.e., a pattern, we first search sequentially through the covering array until finding the first row whose element at the $i$-th position with $s_i = 1$ is equal to $w_i$, for $i = 1, \ldots, n$, and use this index as the encoder output $m$, thus defining the encoder mapping $g$. Since a covering array, by definition, covers all possible patterns, encoding is bound to output an index $m$, and decoding is bound to ensure that whenever $s_i = 1$, $\hat{w}_i = w_i$.

    We say that a pattern $(\uline{s}, \uline{w})$ is ``covered'' by a row in a covering array if the element at the $i$-th position with $s_i = 1$ of the row is equal to $w_i$, for $i = 1, \ldots, n$. Clearly, what the encoder does, as described in the last paragraph, is to find the first row in $\mathsf{CA}(M; k, n, q)$ that covers the pattern for encoding.

    To complete the construction of the DMRA codebook $\mathcal{C}$, we further use a uniquely decodable lossless code, possibly variable-length (e.g., Shannon code), to encode $g(\uline{S}, \uline{W})$. That exploits the possibly non-uniform probability distribution over $g(\uline{S}, \uline{W})$ and substantially reduces the expected codeword length, as will be revealed subsequently.

    Now let us address the key task of constructing the desired $\mathsf{CA}(M; k, n, q)$. We study a greedy strategy, as adopted in \cite{cohen97aetg, bryce07density}. We start with an empty array, covering no pattern, and therefore the number of uncovered patterns is initialized as $U_0 = \binom{n}{k} \cdot q^k$. Then we add a vector from $\mathcal{Q}^n = \{0, \ldots, q - 1\}^n$ as the first row in $\mathsf{CA}(M; k, n, q)$, denoted as $\mathsf{CA}_1$ for short. Any such vector in $\mathcal{Q}^n$ clearly covers $\binom{n}{k}$ patterns, and is hence arbitrarily chosen due to the greedy nature of strategy. Therefore the number of uncovered patterns after adding $\mathsf{CA}_1$ becomes
    \begin{equation}
        U_1 = U_0 - \binom{n}{k} = U_0 (1 - q^{-k}).
    \end{equation}
    We iteratively add new (i.e., not appearing yet) vectors to the covering array, in a greedy fashion that each newly added vector covers as many previously uncovered patterns as possible. That is, if the covering array already contains $\{\mathsf{CA}_1, \ldots, \mathsf{CA}_m\}$, then we choose $\mathsf{CA}_{m + 1}$ as the one in $\mathcal{Q}^n \setminus \{\mathsf{CA}_1, \ldots, \mathsf{CA}_m\}$ that covers the largest number of patterns uncovered by any of $\{\mathsf{CA}_1, \ldots, \mathsf{CA}_m\}$.

    A basic relationship observed in \cite{cohen97aetg, bryce07density} reveals a super-geometric decay along $U_0, \ldots, U_m, \ldots$, numbers of uncovered patterns in consecutive additions of vectors to the covering array. It is given as follows:
    \begin{equation}
        U_{m+1} \leq U_m-\frac{U_m}{q^k} = U_m(1-q^{-k}), \quad m = 0, 1, \ldots\label{recursion}
    \end{equation}

    At this point, let us use (\ref{recursion}) to complete the proof of Theorem \ref{thm:main}, and for the sake of completeness, provide a proof of (\ref{recursion}) after that.

    Since the random pattern $(\uline{S}, \uline{W})$ is uniformly distributed according to the problem formulation in Section \ref{sec:pro.}, evaluating the probability distribution of $g(\uline{S}, \uline{W})$ reduces into simple counting:
    \begin{equation}
        P(m) = \mathrm{Pr}\{g(\uline{S}, \uline{W}) = m\} = \frac{U_{m - 1} - U_m}{U_0}, \quad m = 1, 2, \ldots
    \end{equation}

    Therefore, we can bound the expectation of the encoder output index $g(\uline{S}, \uline{W})$:
    \begin{align}
        \mathbb{E}[g(\uline{S}, \uline{W})] 
        &= \sum_m P(m) m \nonumber\\
        &= \frac{1}{U_0} \sum_m (U_{m - 1} - U_m) m \nonumber\\
        &= \frac{1}{U_0} \sum_m U_m \nonumber\\
        &\leq \frac{1}{U_0} \sum_m U_0 (1 - q^{-k})^m = q^k,
    \end{align}
    where the inequality is from recursively applying (\ref{recursion}).

    For probability distributions supported over positive integers with a given expectation, geometric distribution achieves the maximum entropy. Therefore, we have
    \begin{align}
        H(g(\uline{S}, \uline{W})) &\leq H(\mathrm{GEO}(1/\mathbb{E}[g(\uline{S}, \uline{W})])) \nonumber\\
        &\leq H(\mathrm{GEO}(q^{-k})) \nonumber\\
        &= k \log q - (q^k - 1) \log (1 - q^{-k}) \nonumber\\
        &\leq k \log q + \log e,
    \end{align}
    where we have used the fact that the entropy of geometric distribution is monotonically increasing with its expectation.
    
    Since a uniquely decodable lossless code enforces an integer constraint, we have that the expected codeword length satisfies
    \begin{align}
        \bar{\ell} &= \mathbb{E}[\ell(g(\uline{S}, \uline{W}))] \nonumber\\
        &< H(g(\uline{S}, \uline{W})) + 1 \nonumber\\
        &\leq k \log q + 1 + \log e,
    \end{align}
    establishing (\ref{var.bound}), which can be guaranteed by the Shannon code, for example.

    Now let us provide a proof of (\ref{recursion}), whose idea essentially follows that in \cite[Sec. 3]{cohen97aetg}. Define the set:
    \begin{equation}
        \mathcal{U} = \left\{(v, p): v \in \mathcal{Q}^n, p \text{ covered by } v \right\};
    \end{equation}
    that is, the collection of all possible vectors in a covering array and all the patterns covered correspondingly.

    Suppose that $\mathsf{CA}_1, \ldots, \mathsf{CA}_m$ have been added already. Define the set:
    \begin{align}
        \mathcal{U}^\prime &= \left\{(v, p): v \in \mathcal{Q}^n \setminus \{\mathsf{CA}_1, \ldots, \mathsf{CA}_m\}, \right. \nonumber\\
        & p \text{ covered by } v \text{ but uncovered by any of } \nonumber\\
        & \left. \{\mathsf{CA}_1, \ldots, \mathsf{CA}_m\} \right\};
    \end{align}
    that is, the subset of $\mathcal{U}$ that have not been covered by $\mathsf{CA}_1, \ldots, \mathsf{CA}_m$.

    Then we calculate the cardinality of $\mathcal{U}^\prime$ in two ways:
    \begin{itemize}
        \item For any $v \in \mathcal{Q}^n$, denote by $n(v)$ the number of patterns that have not been covered by $\mathsf{CA}_1, \ldots, \mathsf{CA}_m$ but can be covered by $v$. These are exactly what $\mathcal{U}^\prime$ describes, and therefore,
        \begin{equation}\label{eqn:Uprime1}
            |\mathcal{U}^\prime| = \sum_v n(v).
        \end{equation}
        If $v \in \{\mathsf{CA}_1, \ldots, \mathsf{CA}_m\}$ then obviously $n(v) = 0$, so that these do not affect the calculation.
        \item Note that for any pattern $p$, there are exactly $q^{n - k}$ vectors in $\mathcal{Q}^n$ that cover it because only $k$ positions in a covering vector are specified by $p$. For any $(v, p)$ in $\mathcal{U}^\prime$, all the $q^{n - k}$ covering vectors for $p$ should be contained in $\mathcal{U}^\prime$, because otherwise such a $(v, p)$ would be placed in $\mathcal{U} \setminus \mathcal{U}^\prime$ instead. Furthermore, the number of patterns in $\mathcal{U}^\prime$ is nothing but $U_m$. Therefore we have
        \begin{equation}\label{eqn:Uprime2}
            |\mathcal{U}^\prime| = q^{n - k} U_m.
        \end{equation}
    \end{itemize}

    Comparing (\ref{eqn:Uprime1}) and (\ref{eqn:Uprime2}), we thus have
    \begin{align}\label{eqn:Uprime1vs2}
        q^{n - k} U_m = \sum_v n(v) \leq q^n \max_v n(v).
    \end{align}

    The greedy strategy finds a vector $v$ that attains $\max_v n(v)$ as $\mathsf{CA}_{m + 1}$ and adds it to the covering array. When there are more than one such vectors, choose any of them arbitrarily. This way, $\max_v n(v)$ is nothing but $U_m - U_{m + 1}$, and consequently (\ref{eqn:Uprime1vs2}) leads to
    \begin{align}
        q^{n - k} U_m &\leq q^n (U_m - U_{m + 1}) \nonumber\\
        q^{-k} U_m &\leq U_m - U_{m + 1},
    \end{align}
    which is exactly (\ref{recursion}).
\end{proof}

\subsubsection{Deterministic versus Random Coding}

In Theorem \ref{thm:main}, we establish that it is possible to reduce the overhead of DMRA coding to within $1 + \log_2 e$ bits, irrespective of the number of total users $n$. This is consistent with \cite{song25tit}. We, however, emphasize that our approach does not require any random coding argument. In fact, our proof as an achievability scheme specifies what a deterministic DMRA code can achieve, for any finite values of $(n, k, q)$, and such a deterministic DMRA code can be readily constructed via a greedy strategy; --- see \cite{cohen97aetg, bryce07density} for its algorithmic implementation.

\subsubsection{Bit-by-bit Coding}

The cost of constructing covering arrays increases rapidly with the message alphabet size $q$. This motivates us to adopt a simple bit-by-bit coding scheme when $q = 2^r$ for some integer $r$. We split each user's message into $r$ bits, and use a binary covering array to output an index for each bit. The thus obtained $r$ indices, which are mutually independent, are then jointly encoded by an overall lossless code.

In spirit of Theorem \ref{thm:main}, the expected codeword length of this bit-by-bit coding scheme is upper bounded by $\bar{\ell} < k r + 1 + r \log e$, because the one bit overhead due to integer-valued codeword length assignment can be amortized among $r$ bits. In contrast, directly encoding over size-$q$ messages yields $\bar{\ell} < k r + 1 + \log e$. So we see that there is a typically modest efficiency loss due to bit-by-bit coding, but the overhead is still independent of $n$.

\subsubsection{Fixed-length Coding}

A useful observation in \cite{cohen97aetg} regarding the greedy strategy of constructing covering arrays is that the number of rows, $M$, in $\mathsf{CA}(M; k, n, q)$ is a finite number, on the order of $\log n$. This can be readily shown by noting that $U_m$ as the number of uncovered patterns after adding $\mathsf{CA}_m$ must be an integer. Iteratively applying (\ref{recursion}) yields:
\begin{equation}
    U_m \leq U_0 (1 - q^{-k})^m.
\end{equation}
Therefore, whenever $U_0 (1 - q^{-k})^m < 1$, it follows that $U_m = 0$. This leads to an upper bound on $M$ as
\begin{equation}
    M \leq \frac{\log \binom{n}{k} + k \log q}{\log \frac{1}{1 - q^{-k}}}, 
\end{equation}
which grows logarithmically with $n$, for fixed $k$ and $q$.

Indeed, such a $\log n$ growth is also necessary for $M$. This can be argued by noting that any two columns in a covering array cannot be exactly identical. Hence the pigeon hole principle indicates that $q^M > n$, leading to $M > \log n /\log q$. For more refined estimates of $M$, see, e.g., \cite{sarkar17siam}.

Since $M$ is finite, in the DMRA code we can use the Huffman algorithm to attain the most efficient representation of $g(\uline{S}, \uline{W})$. On the other hand, since $M$ is on the order of $\log n$, in some applications we may even skip the variable-length lossless coding step and directly use a fixed-length representation of $g(\uline{S}, \uline{W})$, attaining $\bar{\ell}$ on the order of $\log \log n$, which is still quite modest over a wide range of $n$.

This $\log\log n$ overhead is similar to the minimum feedback rate for collision‑free scheduling in massive random access \cite{kang21tit}. The common underlying reason in both settings is that a combinatorial structure with only logarithmic size in $n$, namely a covering array herein or a perfect hashing family in \cite{kang21tit}, suffices to resolve the uncertainty about which users are active, leading to an overhead that grows very slowly with the system scale.

\section{Numerical Experiments}\label{sec:numerical}

We present some preliminary numerical experiments to corroborate the analysis in the previous section. As established in the proof of Theorem \ref{thm:main}, the construction of the DMRA code is based on an iterative greedy strategy for constructing covering arrays, ensuring that in each iteration the newly added vector covers as many previously uncovered patterns as possible. For efficient implementation, we employ the deterministic density algorithm \cite{bryce07density}, which guarantees that each added vector covers at least the average number of uncovered patterns. This is not strictly greedy but is still sufficient to enforce the number of uncovered patterns to decrease at least geometrically with the index. For simplicity, the numerical experiments primarily focus on binary DMRA codebooks, i.e., $q = 2$.

\begin{figure}[htbp]
   \centering
   \subfigure[$k=2$]{
       \includegraphics[width=0.47\textwidth]{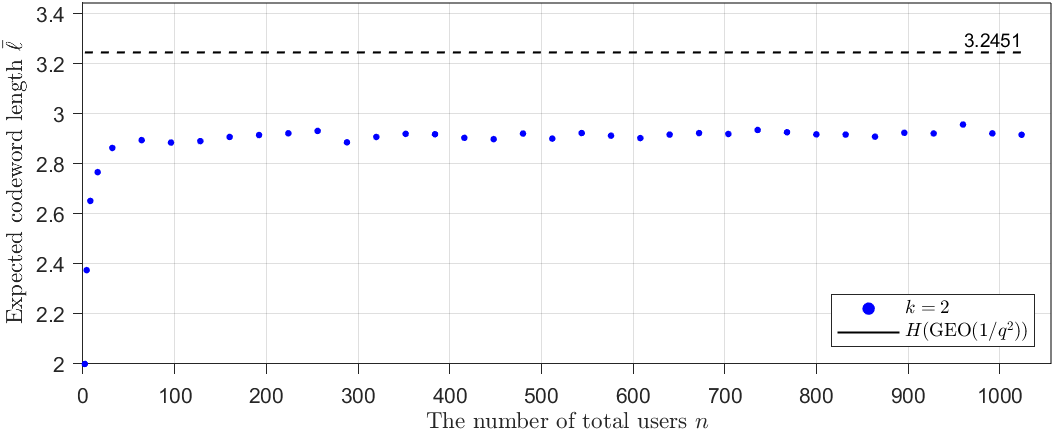}
   }\\
   \subfigure[$k=3$]{
       \includegraphics[width=0.47\textwidth]{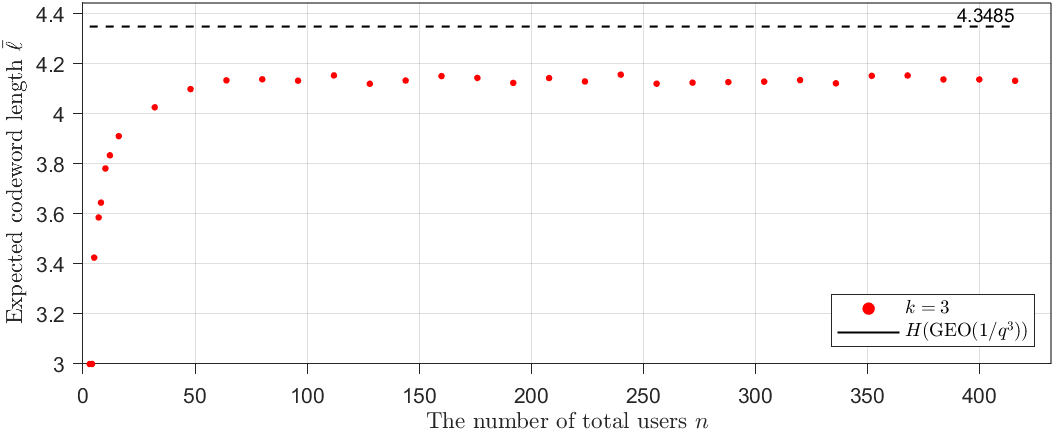}
   }
   \caption{Expected codeword length $\bar{\ell}$ vs number of total users $n$}
   \label{fig:varlength}
\end{figure}

To minimize the expected codeword length, we encode the selected index of the covering array with a uniquely decodable lossless code, in particular, the Huffman code, as described in the proof of Theorem \ref{thm:main}.

The resulting expected codeword lengths are displayed in Figure \ref{fig:varlength}, for $k = 2$ (upper) and $3$ (lower). As shown, the expected codeword length $\bar{\ell}$ in dots slightly fluctuates with $n$, due to the integer effect of Huffman coding, but still largely remains stable without any trend of growth, as $n$ increases.

We can read from Figure \ref{fig:varlength} that for large $n$, the overhead beyond $k$ bits (the net payload of active users) is within $1$ and $1.2$ bits for $k = 2$ and $3$, respectively.

We also plot $H(\mathrm{GEO}(q^{-k}))$, the entropy of the geometric distribution with expectation $q^k$, in dashed horizontal lines, with its numerical value indicated as well. There exists a noticeable gap between $\bar{\ell}$ and $H(\mathrm{GEO}(q^{-k}))$, confirming the super-geometric decay of uncovered patterns across the indices of covering array.

\section{Conclusion}\label{sec:conclusion}

The DMRA model, a massive access setup describing the scenario where a base station sends messages to sporadic users, is studied via the lens of covering arrays. A general upper bound of $1 + \log_2 e$ bits independent of the number of total users is established, on the coding overhead due to the lack of global knowledge of user activity at individual users. The analysis is based on deterministic codes, without requiring ensemble analysis over random codes. Furthermore, the deterministic codes can be constructed via greedy search.

For future research, of particular importance and interest is exploring more efficient and systematic construction of covering arrays, suitable for the needs of DMRA coding. Greedy algorithms are computationally expensive and quickly become infeasible as $n$ or $q$ increases. The lack of structure in covering arrays constructed by greedy algorithms also demands high storage cost, and hinders efficient addressing for encoding and decoding. Resolving these challenges will be necessary for paving the way towards practical applications of DMRA in future massive access systems.


\bibliographystyle{IEEEtran}
\bibliography{reference}

\end{document}